\documentclass[a4paper]{article}
\usepackage[textheight=234mm,textwidth=151mm]{geometry}
\usepackage{amsthm,amssymb,graphicx}

\newtheorem{pros}{Proposition}

\newtheorem{defn}{Definition}

\title{\bf\sffamily A Mathematically Sensible Explanation of the Concept of Statistical Population}
\author{Yiping Cheng\\ Beijing Jiaotong
University, China\\\texttt{ypcheng@bjtu.edu.cn}}
\date{}

\begin{document}
\vskip -16pt \maketitle

\begin{abstract}In statistics education, the concept of population is widely felt hard to
grasp, as a result of vague explanations in textbooks. Some
textbook authors therefore chose not to mention it. This paper
offers a new explanation by proposing a new theoretical framework
of population and sampling, which aims to achieve high
mathematical sensibleness. In the explanation, the term population
is given clear definition, and the relationship between simple
random sampling and iid random variables are examined
mathematically.
\end{abstract}


\section{Introduction}

The theory of statistics is based on the theory of probability,
and the first mathematical concept of statistics beyond pure
probability is that of {\em random samples\/}. A closely related
semi-mathematical concept, namely that of {\em statistical
population\/}, is also popular among statistical practitioners.
The term {\em population\/} is usually used instead of {\em
statistical population\/} where no confusion would arise.

However, despite its popularity in the statistical community, in
textbooks the treatment of the concept of population is quite
varied, which reflects varied views of different authors on that
concept. As a result of our study of 9 undergraduate-level
textbooks in statistics:
\cite{groot89,nitis00,DKLM05,devore12,hogg13,montgomery03,soong04,wackerly08,walpole12},
we found that they can be classified into two groups according to
the way they introduce the concept of random samples. Let us give
the detail below.

\subsection{Group A: Textbooks That Avoid Mentioning Population}

This group consists of, in chronological order,
\cite{groot89,nitis00,DKLM05,devore12,hogg13}. One can check
\cite[page 145]{groot89}, \cite[page 246]{DKLM05}, and \cite[page
204]{hogg13} for evidences of our classification. Book
\cite{nitis00} does not explicitly give a formal definition of
random samples although it does use the term, and book
\cite{devore12} uses the term of random samples very early, with
its definition deferred until page 214.

In \cite[page 204]{hogg13}, a definition of random samples is
given as follows: \vskip 6pt

\fbox{\begin{minipage}{5.4in}\small \hskip 1em {\bf Definition:\/}
If the random variables $X_1,X_2, \cdots,X_n$ are independent and
identically distributed (iid), then these random variables
constitute a random sample of size $n$ from the common
distribution.\end{minipage}}

\vskip 6pt

Definitions of this concept in other textbooks in this group,
though with different wordings, are mathematically equivalent, so
they are omitted here. By defining random samples directly from
distribution without mentioning population, the authors obviated
the difficult task of explaining the population concept.

\subsection{Group B: Textbooks That Use Population to Introduce Random Samples}

This group consists of, in chronological order,
\cite{montgomery03,soong04,wackerly08,walpole12}. Of them, we
found in \cite{montgomery03,walpole12} detailed explanations of
the concepts of population and random samples before giving formal
definitions. We now quote the following excerpt of the explanation
in \cite[pages 195--197]{montgomery03}, which is similar to the
explanation in \cite[page 225--227]{walpole12}.

 \vskip 6pt\fbox{\begin{minipage}{5.4in}\small \hskip 1em A population
consists of the totality of the observations with which we are
concerned.

\hskip 1em In any particular problem, the population may be small,
large but finite, or infinite. The number of observations in the
population is called the size of the population. For example, the
number of underfilled bottles produced on one day by a soft-drink
company is a population of finite size. The observations obtained
by measuring the carbon monoxide level every day is a population
of infinite size. We often use a probability distribution as a
model for a population. For example, a structural engineer might
consider the population of tensile strengths of a chassis
structural element to be normally distributed with mean $\mu$ and
variance $\sigma^2$. We could refer to this as a normal population
or a normally distributed population.

\hskip 1em A sample is a subset of observations selected from a
population.

\hskip 1em To define a random sample, let $X$ be a random variable
that represents the result of one selection of an observation from
the population. Let $f(x)$ denote the probability density function
of $X$. Suppose that each observation in the sample is obtained
independently, under unchanging conditions. That is, the
observations for the sample are obtained by observing $X$
independently under unchanging conditions, say, $n$ times. Let
$X_i$ denote the random variable that represents the $i$-th
replicate. Then, $X_1,X_2, \cdots,X_n$ is a random sample and the
numerical values obtained are denoted as $x_1,x_2, \cdots,x_n$.
The random variables in a random sample are independent with the
same probability distribution $f(x)$ because of the identical
conditions under which each observation is obtained.
\end{minipage}}  \vskip 6pt

In \cite[page 259]{soong04} there is a short reference to the
concept of population.

\vskip 6pt \fbox{\begin{minipage}{5.4in}\small \hskip 1em The
appropriate representation of $\hat{\theta}$ is
$\hat{\Theta}=h(X_1,X_2, \cdots,X_n)$ where $X_1,X_2, \cdots,X_n$
are random variables representing a sample from random variable
$X$, which is referred to in this context as the {\em
population\/}.
\end{minipage}}  \vskip 6pt

And in \cite[page 78]{wackerly08} there is concise
characterization of random sampling.

\vskip 6pt\fbox{\begin{minipage}{5.4in}\small \hskip 1em Let $N$
and $n$ represent the numbers of elements in the population and
sample, respectively. If the sampling is conducted in such a way
that each of the $C_N^n$ samples has an equal probability of being
selected, the sampling is said to be random, and the result is
said to be a random sample.
\end{minipage}}  \vskip 6pt

\subsection{Motivation for a Mathematically Sensible Explanation}

The concept of population is widely felt hard to grasp by
students, which is perhaps the primary reason for it not being
included in many textbooks. We are not attempting to argue whether
we should or should not include it in textbooks. Instead, we try
to propose an explanation of the population concept that makes
perfect sense in a mathematical manner (which we call {\em
mathematically sensible\/} in the sequel). Our attempt was
motivated by our perception that the existing explanations (the
foregoing explanation of \cite{montgomery03} being a good
representative) of the concept are not adequately mathematically
sensible. Our reasons are as follows:
\begin{enumerate}
\item One would expect in a mathematically sensible explanation,
the population fits into the probability theory framework in that
it corresponds to a clear-cut mathematical concept. The existing
explanations are chaotic in this respect. Let us examine each of
the three possibilities:
\begin{enumerate} \item {\bf population=sample space\/}. This
possibility accords with the first half of the second paragraph of
the foregoing quote of \cite[pages 195--197]{montgomery03}.
However, how can a sample space have a distribution? Only random
variables can. \item {\bf population=range space of the random
variable\/}. But this possibility directly contradicts the second
example given in the first half of the second paragraph of the
above-mentioned quote. If the carbon monoxide level can have only
100 values due to quantization, then should we say the size of
population is now 100? \item {\bf population=the random
variable\/}. This possibility is the most-often seen. But then
what is the corresponding mathematical concept for observation?
How can we view a population as a set of observations, as now the
population is a random variable? The first sentence of the last
paragraph of the above-mentioned quote suggests one random
variable corresponds to one observation, which is in contradiction
to the first sentence of the quote under this possibility.
\end{enumerate}
\item The existing explanations do not mathematically show why
random sampling in its primitive sense (characterized by the
foregoing quote of \cite[page 78]{wackerly08}) leads to
independent and identically distributed random variables.
\end{enumerate}

The question now becomes: Does a mathematically sensible
explanation exist? Our answer is yes and we will provide such an
explanation in this paper.

\section{A New Probability Framework for Population and Sampling}

Actually, our new explanation of the population concept is a part
of a new probability framework for population and sampling. Before
giving the detail, let us look at its salient features:
\begin{enumerate}
\item It involves two probability spaces: a {\em population
probability space\/} and an {\em experiment probability space\/}.
However they are closely related to each other. \item The term
population is no longer mathematically overloaded as in the
existing explanations, it now corresponds to the clear-cut
mathematical concept of the sample space of the population
probability space. \item The random variable $X$, which was
regarded as the population random variable in the existing
explanations, lives in the population probability space, whereas
the sample $X_1,X_2, \cdots,X_n$ lives in the experiment
probability space. So mathematically speaking, they are NOT
defined over the same probability space.\item The physical
operation of sampling is described by the mathematical operation
of mapping from the experiment space (i.e. sample space of the
experiment probability space) to the population, and the {\em
simple randomness\/} of the the sampling is mathematically
described by the to-be-defined {\em simpleness\/} of the mapping.
\end{enumerate}

\subsection{The Population Probability Space}

Recall that a probability space is a triple
$(\Omega,\mathcal{F},P)$,  where $\Omega$ is a set called the
sample space, $\mathcal{F}$ is a $\sigma$-field of subsets of
$\Omega$ (see \cite[page 10]{hogg13} for a description), and
$P:\mathcal{F}\mapsto [0,1]$ is a probability function as defined
in \cite[page 11]{hogg13}.

Mathematically speaking, any probability space can be a population
probability space. It is only that the sample space of a
population probability space is called population, which,
understood physically, consists of individuals with various
observable quantities. For example, if we are carrying out a
health survey in a state with one million people, then we can
define a population probability space as follows.
\begin{eqnarray*}
\Pi &=& \{0,1,\cdots,999999\}, \\
\mathcal{F}_\Pi &=& \{\mbox{All subsets of }\Pi\},\\
P_\Pi(A)&=&{\mbox{Number of elements in }A\over 1000000},\mbox{
for any}A\subseteq \Pi.
\end{eqnarray*}

Obviously, $(\Pi,\mathcal{F}_\Pi,P_\Pi)$ constitutes a probability
space. The elements in $\Pi$ are ordinal ids for the people in the
state. We can then define random variables $X$, $Y$, $Z$, ... to
represent the length, weight, blood type, etc. of the people. For
example, it can be so defined that $X(785)$ is the length in
meters of the individual whose ordinal id is 785 in this state.
Note also that if the sample space is finite and all its elements
are assigned equal probability, then that probability space is
called ``classical" in the literature. So this example is
classical. However, our theory applies also to non-classical and
infinite population cases.

The population probability space can be understood either
stochastically or non-stochastically. It can be understood
non-stochastically, since in the above description, no physical
randomness is involved. It can also be understood stochastically
if we view $P_\Pi(A)$ as the the ``(physical) probability of any
member of $A$ being selected for survey".

\subsection{The Experiment Probability Space and Simple Sampler Mapping}

An experiment probability space is one where the sample space is
the set of experiment outcomes. We denote it by the triple
$(E,\mathcal{F}_E,P_E)$.

In our framework, size-$n$ sampling is selection of an $n$-tuple
from the population. What $n$-tuple is selected is determined by
the experiment outcome. Therefore, sampling can be mathematically
described by a {\em sampler mapping\/} $S:E\mapsto \Pi^n$.
Practically, it is the usual case that the size of population is
massively larger than $n$, and thus one element of the tuple can
rarely equal another element. In addition, it is also practically
usual that the order of the elements does not matter in later
statistical processing. Therefore practically we often say
sampling is selection of a subset of size $n$ from the population.
The $n$-tuple term is adopted here, because it is considered more
convenient.

Now we give the mathematical property to capture the simple
randomness as described in the quote of \cite[page 78]{wackerly08}
in Section 1.

\begin{defn}
Let $(\Pi,\mathcal{F}_\Pi,P_\Pi)$ be a population probability
space, and $(E,\mathcal{F}_E,P_E)$ be an experiment probability
space. A sampler mapping $S:E\mapsto \Pi^n$ is said to be simple,
if for all
$B_1\in\mathcal{F}_\Pi,B_2\in\mathcal{F}_\Pi,\cdots,B_n\in\mathcal{F}_\Pi$,
\begin{equation} P_E(\{e|S_1(e)\in B_1,\  S_2(e)\in B_2,\ \cdots,\
S_n(e)\in B_n\})=P_\Pi(B_1)P_\Pi(B_2)\cdots P_\Pi(B_n).
\end{equation}
\end{defn}

For a fixed experiment probability space, the existence of a
simple sampler mapping is not guaranteed. In fact, it is more a
requirement on the experiment probability space than one on the
sampler mapping.

The quote of \cite[page 78]{wackerly08} is with an implicit
assumption which, put in our terminology, is the population
probability space being classical. Then, a subset of $\Pi$ of size
$n$ corresponds to $n!$ tuples with different orders, and if the
mapping is simple, each order has $1\over N^n$ probability of
being selected, thus the probability of the subset being selected
is ${n!\over N^n}$, irrespective of what elements constitute the
subset. This is just the simple-randomness of sampling.

\subsection{Mathematical Construction of Experiment Probability Space and
Simple Sampler Mapping}

Theoretically, there always exists a trivial construction of
experiment probability space and simple sampler mapping, which is
the product measure space $(\Pi,\mathcal{F}_\Pi,P_\Pi)^n$ and the
identity mapping. But this construction is only of theoretic
value, not expected to yield real benefits.

If we want to construct an experiment probability space and a
simple sampler mapping that faithfully describe a practically used
sampling procedure, then this task is formidable, because in
practice, ad-hoc methods are usually used, which are only
approximately simple random. In such cases we can just assume the
sampling to be simple random and assume the existence of an
experiment probability space and simple sampler mapping, without
giving the particular construction.

Perhaps the only scenario that we want to do the explicit
mathematical construction is when we need the construction to
design a simple random sampling algorithm. In that scenario, the
experiment probability space, viewed as a model for random
generator, should already have an implementation available. The
following random number generators are considered to have reliable
(albeit just close approximate) algorithm implementations:
\begin{description} \item[Discrete:] Given a length $L$,
random-generate $\xi$ from $\{0,1,\cdots,L-1\}$ with each of
$0,1,\cdots,L-1$ having equal probability. \item[Continuous:]
Uniformly random-generate $\xi$ from real interval $[0,1]$.
\end{description}

\subsubsection{Construction Based on the Discrete Generator}

This method is applicable only when the population probability
space is classical. So assumed, let $N$ be the size of population,
and let
\begin{eqnarray}
E &=& \{0,1,\cdots,N^n-1\}, \\
\mathcal{F}_E &=& \{\mbox{All subsets of }E\},\\
P_E(A)&=&{\mbox{Number of elements in }A\over N^n},\mbox{ for
any}A\subseteq E.
\end{eqnarray}

Eqs. (2--4) constitute a model for the length-$N^n$ discrete
generator. Now we define the sampler mapping: \begin{equation}
S(e)=(a_1,a_2,\cdots,a_n)
\end{equation} where \begin{equation} e=a_1
N^{n-1}+a_2N^{n-2}+\cdots+a_n\end{equation}with
\begin{equation}0\leq a_1\leq N-1,\quad 0\leq a_2\leq N-1,\quad\cdots,\quad 0\leq a_n\leq
N-1.\end{equation}

It is easy to see that $S$ is a simple sampler mapping.

\subsubsection{Construction Based on the Continuous Generator}

This method is applicable when the population probability space is
finite, both classical and nonclassical. However, we only describe
the method for the classical case, as the nonclassical case is
much more complicated to deal with. Let $N$ be the size of
population, and let
\begin{eqnarray}
E &=& [0,1], \\
\mathcal{F}_E &=& \{\mbox{All Borel subsets of }E\},\\
P_E(A)&=&\mbox{Measure of }A,\mbox{ for any }A\in \mathcal{F}_E.
\end{eqnarray}

Eqs. (8--10) constitute a model for the continuous generator. Now
we define the sampler mapping: \begin{equation}
S(e)=(a_1,a_2,\cdots,a_n)
\end{equation} where \begin{equation} e=a_1
N^{-1}+a_2N^{-2}+\cdots+a_nN^{-n}+\cdots\end{equation}with
\begin{equation}0\leq a_1\leq N-1,\quad 0\leq a_2\leq N-1,\quad\cdots,\quad 0\leq a_n\leq
N-1.\end{equation}

We omit the proof to show that $S$ is a simple sampler mapping
with the classical assumption.

\subsection{Why Simple Random Sampling Leads to IID Random Variables}

In Section 2.2 we showed that in the classical case, simpleness of
the sampler mapping over the experiment probability space implies
simple randomness of sampling. However, the reverse implication
may not hold, and we may need a stronger version. Thus we try to
revise the quote of \cite[page 78]{wackerly08} a bit by replacing
``subsets" with ``tuples":

\vskip 6pt\fbox{\begin{minipage}{5.4in}\small \hskip 1em Let $N$
and $n$ represent the numbers of elements in the population and
sample, respectively. If the sampling is conducted in such a way
that each of the $N^n$ ordered samples has an equal probability of
being selected, the sampling is said to be simple random, and the
result is said to be a random sample.
\end{minipage}}  \vskip 6pt

This tuple-version is still completely intuitively reasonable and
is readily seen to be equivalent in the classical case to the
simpleness of the sampler mapping over the experiment probability
space. So now this version of simple randomness will be adopted.
Furthermore, because of their equivalence in the classical case,
and the simpleness concept is more general, we decide to use the
simpleness as defined in Definition 1 to carry out our derivation
of the following proposition about random variables.
\begin{pros}
Let $X$ be a random variable over population probability space
$(\Pi,\mathcal{F}_\Pi,P_\Pi)$. Let $(E,\mathcal{F}_E,P_E)$ be an
experiment probability space and $S:E\mapsto \Pi^n$ be a simple
sampler mapping. Then \begin{enumerate} \item For $i=1,\cdots,n$,
$X_i$ as defined by $X_i(e):=X(S_i(e))$ is a random variable over
$(E,\mathcal{F}_E,P_E)$ and has the same distribution function as
$X$. \item $X_1,X_2,\cdots,X_n$ are independent.
\end{enumerate}
\end{pros}

\begin{proof}
1. Fix an arbitrary $u\in\mathbb{R}$. Let $B=\{\pi|X(\pi)<u\}$,
then since $X$ is a random variable over
$(\Pi,\mathcal{F}_\Pi,P_\Pi)$, we have $B\in\mathcal{F}_\Pi$. Then
\begin{equation}\{e|X_i(e)<u\}=\{e|X(S_i(e))<u\}=\{e|S_i(e)\in
B, S_j(e)\in \Pi\mbox{ for }j\neq i\}. \end{equation} The right
side of (14) is in $\mathcal{F}_E$, because it has its probability
defined by Definition 1. Since $u$ is arbitrary, we have $X_i$ is
a random variable over $(E,\mathcal{F}_E,P_E)$. Furthermore, also
by Definition 1,
\[P_E(\{e|X_i(e)<u\})=P_E(\{e|S_i(e)\in
B, S_j(e)\in \Pi\mbox{ for }j\neq
i\})=P_\Pi(B)=P_\Pi(\{\pi|X(\pi)<u\}).
\] Since $u$ is arbitrary, this means $X_i$ has the same
distribution function has $X$.

2. Fix an arbitrary set of real values $u_1,u_2,\cdots,u_n$. Then
\[P_E(\{e|X_1(e)<u_1,X_2(e)<u_2,\cdots,X_n(e)<u_n\})\]
\[=P_E(\{e|X(S_1(e))<u_1,X(S_2(e))<u_2,\cdots,X(S_n(e))<u_n\})\]
\[=P_E(\{e|S_1(e)\in\{\pi|X(\pi)<u_1\},S_2(e)\in\{\pi|X(\pi)<u_2\},\cdots,S_n(e)\in\{\pi|X(\pi)<u_n\}\})\]
\[=P_\Pi(\{\pi|X(\pi)<u_1\})P_\Pi(\{\pi|X(\pi)<u_2\})\cdots P_\Pi(\{\pi|X(\pi)<u_n\})\]
\[=P_E(\{e|X_1(e)<u_1\})P_E(\{e|X_2(e)<u_2\})\cdots
P_E(\{e|X_n(e)<u_n\}).\] This means $X_1,X_2,\cdots,X_n$ are
independent.
\end{proof}

\section{Concluding Remarks}

Section 2 presents a fully developed theoretical framework for
population and sampling. A large portion of the previous
explanations in verbal language has now been replaced by results
in mathematical language. This makes our explanation much less
vague than previous ones. However, some of the results may require
a bit too much of mathematical maturity of the reader, and
therefore in courses or textbooks, it may not be the most suitable
to present the framework in its full version. However, we believe
the following points should be stressed in educational occasions:
\begin{enumerate}
\item The term ``population" as a noun should refer to the sample
space, not the random variable as is the case in many textbooks.

\item The term ``population" can be used as an attributive in
``population random variable", ``population distribution", and
``population density", which refer to a random variable in the
population probability space, its distribution, and its density,
respectively.

\item The population random variable $X$, and the sample random
variables $X_1,X_2,\cdots,X_n$ do not live in the same probability
space. Failure to notice this is the cause for many difficulties
in the existing explanations.

\item The term ``sample" may suggest students to believe the
random sample $X_1,X_2,\cdots,X_n$ contains less information than
the population random variable $X$. Actually, mathematically
speaking, $X_1,X_2,\cdots,X_n$, and even each of its members,
$X_i$, contain no less information than $X$. It is only that in
practice, only one experiment is done, and we only hold $n$ real
values $x_1,x_2,\cdots,x_n$ as observed values of the sample.
Usually $n<<N$.
\end{enumerate}

Finally, it is hoped that this new framework and explanation, with
its unique feature of mathematical sensibleness, will be helpful
to a large number of statistics students.

\end{document}